\documentclass{article}

\usepackage[margin=1in]{geometry}
\usepackage{authblk}
\usepackage{amssymb,verbatim,xcolor,fancyhdr}
\usepackage{amssymb,verbatim,xcolor,fancyhdr,amsmath,amsthm}
\usepackage[utf8]{inputenc} 
\usepackage[T1]{fontenc}    
\usepackage{hyperref}       
\usepackage{url}            
\usepackage{booktabs}       
\usepackage{amsfonts}       
\usepackage{nicefrac}       
\usepackage{microtype}      
\usepackage{lipsum}
\usepackage{graphicx}
\graphicspath{ {./images/} }
\usepackage[english]{babel}
\newtheorem{theorem}{Theorem}
 
\newtheorem{cor}[theorem]{Corollary}

\title{Burning Number for the Points in the Plane\thanks{This work is supported in part by the Natural Sciences and Engineering Research Council of Canada (NSERC).}}

\author{J. Mark Keil} 
\author{Debajyoti Mondal}
\author{Ehsan Moradi}
 
\affil{Department of Computer Science, University of Saskatchewan, Saskatoon, Canada\\
  \texttt{\{mark.keil,d.mondal,ehsan.moradi\}@usask.ca}}

\begin{document}
\maketitle  
\begin{abstract}
The burning process on a graph $G$ starts with a single burnt vertex, and at each subsequent step, burns the neighbors of the currently burnt vertices, as well as one other unburnt vertex. The burning number of $G$ is the smallest number of steps required to burn all the vertices of the graph. In this paper, we examine the problem of computing the burning number in a geometric setting. The input is a set of points $P$ in the Euclidean plane. The burning process starts with a single burnt point, and at each subsequent step, burns all the points that are within a distance of one unit from the currently burnt points and one other unburnt point. The burning number of $P$ is the smallest number of steps required to burn all the points of $P$. We call this variant \emph{point burning}. We consider another variant called \emph{anywhere burning}, where we are allowed to burn any point of the plane. We show that point burning and anywhere burning problems are both NP-complete, but $(2+\varepsilon)$  approximable for every $\varepsilon 
>0$. Moreover, if we put a  restriction on the  number of burning sources that can be used, then the anywhere burning problem becomes NP-hard to approximate within a factor of $\frac{2}{\sqrt{3}}-\varepsilon$. 
\end{abstract}


\section{Introduction} 
Graph burning is a discrete process that propagates fire to burn all the nodes in a graph. In particular, the fire is initiated at a vertex of the graph and at each subsequent step, the fire propagates to the neighbors of the currently burnt vertices and a new unburnt vertex is chosen to initiate a fire. The vertices where we initiate fire are called the \emph{burning sources}. The burning process continues until all the vertices are burnt. The \emph{burning number} of a graph $G$ is the minimum number of steps to burn all its vertices. 
Bonato et al.~\cite{DBLP:conf/waw/BonatoJR14} introduced graph burning as a model of social contagion. The problem is NP-Complete even for simple graphs such as a spider or forest of paths~\cite{Bessy2017}.

In this paper we introduce burning number for the points in the plane. We consider two methods for burning:  \emph{point burning} and \emph{anywhere burning}. Both problems take a set of points $P$ as an input, and seek for the minimum number of steps to burn all points of $P$. 

In the point burning model, we can initiate  fire only at the given points. The burning process starts by burning one given point, and then at each subsequent step, the fire propagates to all unburnt points of the plane that are within one unit of any burnt point of the plane and a new unburnt given point is chosen to initiate the fire.  Figure~\ref{intro}(top) illustrates this model. Note that we may not have an unburnt vertex at the last step. 

In the anywhere burning model, we can start a fire anywhere on the plane, and at each subsequent step, the fire propagates to all unburnt points of the plane that are within one unit of any burnt point of the plane, and a new unburnt point is chosen to initiate the fire. Figure~\ref{intro}(bottom) illustrates this model. 

\begin{figure}[pt]
    \centering
\includegraphics[width=.6\textwidth]{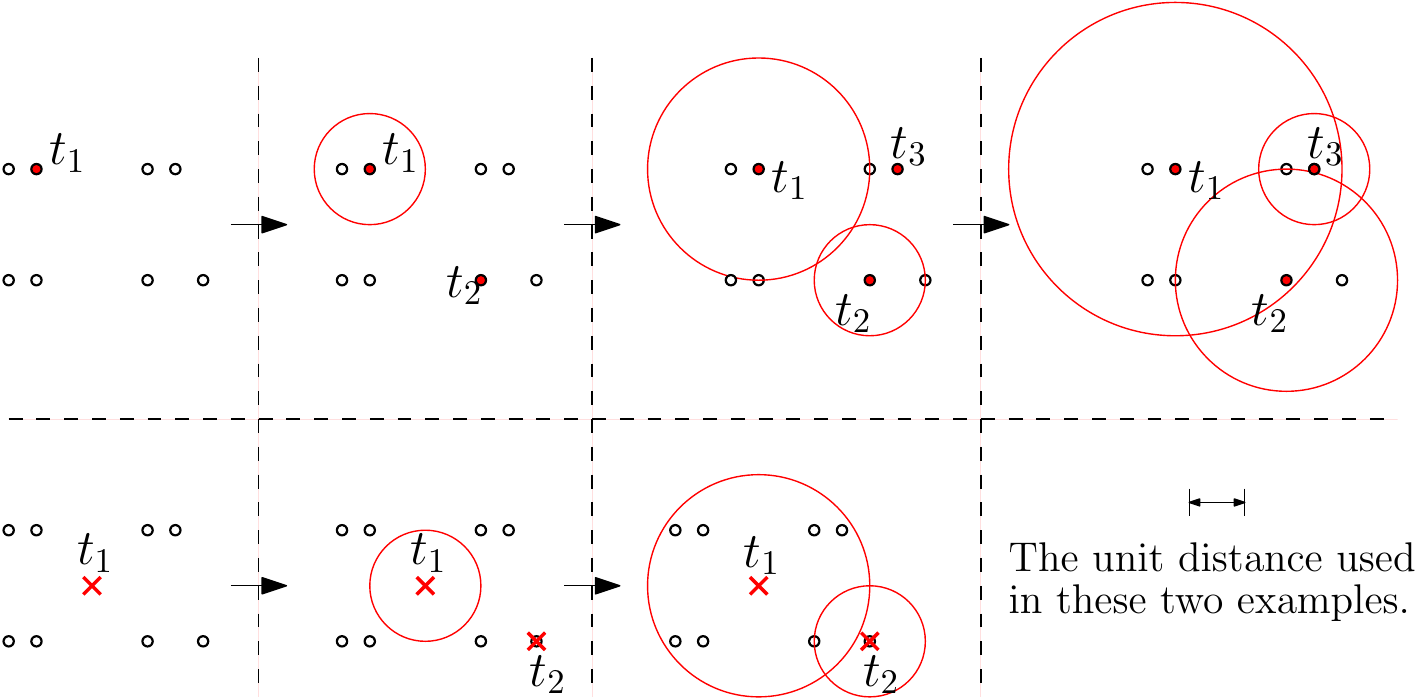}
    \caption{Illustration for (top) point burning and (bottom) anywhere burning. The burning sources are illustrated in labelled dots and cross marks, respectively. }
    \label{intro}
\end{figure}

In addition to being a natural generalization of graph burning,  our proposed burning processes may potentially be used to model supply chain systems. A hypothetical example of how a burning process may model  a  supply chain management system is as follows. Consider a business that needs to maintain a continuous supply of perishable  goods to a set of $P$ locations. Each day it can manage to send one large shipment to a hub location that distributes the goods further to the nearby locations over time. The point burning considers only the points of $P$ as potential hubs, whereas anywhere burning allows to create a hub at any point in the plane. The burning number indicates the minimum number of days needed to distribute the goods to all locations. For example, in Figure~\ref{intro}(top), the hubs are $t_1$, $t_2$, and $t_3$, and  the business can keep sending the shipments to the hubs after  every three days in the same order.

\subsection{Related Results}
{\color{black}

Finding the graph burning number is NP-Hard~\cite{Bessy2017}, but approximable within a factor of 3~\cite{bonato2019approximation}. These results have been improved very recently. Garc\'ia-D\'iaz et al.~\cite{DBLP:journals/access/Garcia-DiazSRC22} have given a $(3-2/b)$-approximation algorithm where  $b$ is the burning number of the input graph. Mondal et al.~\cite{DBLP:conf/walcom/MondalP0R21} have shown the graph burning problem to be APX-hard,  even in a generalized setting where $k=O(1)$ vertices can be chosen to initiate the fire at each step. They gave a 3-approximation algorithm for this generalized  version~\cite{DBLP:conf/walcom/MondalP0R21}. Since the introduction of the graph burning problem~\cite{Bonato2016}, a rich body of literature examines the upper and lower bound on the graph burning number for various classes of graphs~\cite{Sim2017,liu2021burning,das2018burning} as well as the parameterized complexity of computing the burning number~\cite{DBLP:conf/iwpec/KobayashiO20}. We refer the reader to~\cite{DBLP:journals/cdm/Bonato21} for a survey on graph burning.

Researchers have also explored burning number for geometric graphs. Gupta et al.~\cite{DBLP:conf/caldam/GuptaLM21} examined   square grid graphs and gave a 2-approximation algorithm for burning square grids. They also showed the burning number to be NP-Complete for connected interval graphs. Bonato et al.~\cite{DBLP:journals/gc/BonatoGS20} considered the burning process on dynamic graphs, which are growing grids in the Cartesian plane with the center at the origin. They explore the proportion or density of burned vertices relative to the growth speed of the grid. Recently, Evans and Lin~\cite{w2022} have introduced polygon burning, where given a polygonal domain and an integer $k$, the problem seeks for $k$ vertices such that the polygonal domain is burned as quickly as possible when burned simultaneously and uniformly from those $k$ vertices. They gave a 3-approximation algorithm for polygon burning.

 The anywhere burning problem that we introduced can be seen to be related to the nonuniform version of the \emph{$k$-center} problem. Given a set of points, the goal of the $k$-center  problem is to find the minimum radius $R$ and a placement of $k$ disks of radius $R$ to cover all the given points. In the \emph{nonuniform $k$-center} problem~\cite{DBLP:conf/icalp/ChakrabartyGK16}, given a set of points and a set of $k$ numbers ${r_0\ge \ldots \ge r_{k-1}}$, the goal is to find a minimum dilation $\alpha$ and a placement of $k$ disks where the $i$th disk, $1\le i \le k$, has radius $\alpha r_i$ and all the given points are covered. If the anywhere burning number of a set of points is $k$, then the nonuniform $k$-center problem with $r_i = i$ admits a solution with $\alpha=1$. 

}
  
\subsection{Our Contribution}
We introduce two discrete-time processes (i.e., point burning and anywhere burning)  to burn the points in the plane, which naturally extend the graph burning model to the geometric setting.   We prove that in both models,  computing the burning number is NP-hard, and 
give polynomial-time $(2+\varepsilon)$-approximation algorithms. 
We then show that if we put a restriction on the number of burning sources that can be used, then the anywhere burning problem becomes NP-hard to approximate within a factor of   $\frac{2}{\sqrt{3}}-\varepsilon$.

\section{Approximating Burning Number}

\subsection{Point Burning}\label{apppb}

The \emph{burning sources at the $i$th step} are all the vertices that we choose to initiate the fire from the beginning of the burning process to the $i$th step (including the $i$th step). We refer to the number of burning sources as $B_i$. The \emph{maximum burning radius} $R_i$ at the $i$th step of the burning process is the maximum radius over all burning sources. After the $i$th step of the burning process, the maximum burning radius is exactly $(i-1)$ and the number of burning sources is exactly $i$ (except possibly for the last step). Therefore, if $\delta^*$ is the number of steps in the optimal solution, then the number of burning sources is at most $\delta^*$, and the maximum burning radius is exactly $(\delta^*-1)$. Hence for the $i$th step, we have the following.
\begin{equation}\label{eq}
    \delta^* \ge i \ge B_i. 
\end{equation}

\begin{theorem}\label{thmpb}
Given a set $P$ of points in $\mathbb{R}^2$ and an $\varepsilon>0$, one can compute a point burning sequence for $P$ in polynomial time such that the length of the sequence is at most $(2+\varepsilon)$ times the point burning number of $P$.
\end{theorem}
\begin{proof}
Let $G_k$ be a unit disk graph where $k/2$ equals one unit, i.e.,  each vertex of $G_k$  corresponds to a disk of radius $k/2$ in $\mathbb{R}^2$, and there is an edge between two vertices of $G_k$ if their corresponding disks intersect. 
Consider the graph $G_k$ 
on $P$ where $P$ represents the centers of the disks. We denote by $D_k$ a \emph{minimum dominating set} of $G_k$, i.e, the smallest set of vertices such that each vertex of  $G_k$  is either in $D_k$ or a neighbor of a vertex in $D_k$. There exists PTAS  to approximate $D_k$
~\cite{DBLP:conf/waoa/NiebergH05}, i.e., $D_k$ is approximable within a factor of $(1+\varepsilon)$ for every fixed $\varepsilon>0$.

Let $\delta^*$ be the burning number for $P$.  We now claim that  $\delta^*$  must be at least $|D_{\delta^*-1}|$. Suppose for a contradiction that the burning number is strictly smaller than $|D_{\delta^*-1}|$ and let $S$ be the corresponding burning sources. Since the maximum burning radius over $S$ is at most $(\delta^*-1)$, we could use $|S|$ disks, each of radius $(\delta^*-1)$ to burn all the points. Hence, we could choose the disks corresponding to $S$ as a dominating set for $G_{\delta^*-1}$. This contradicts that $D_{\delta^*-1}$ is a minimum dominating set. Hence we have 
    $\delta^*\ge |D_{\delta^*-1}|$.

We now iteratively guess the  burning number $\delta$ from $1$ to $n$, where $|P|=n$. For each $\delta$, we construct $G_{\delta-1}$, and compute  a $(1+\varepsilon)$ approximation $D'_{\delta-1}$ for $D_{\delta-1}$. If $\frac{|D'_{\delta-1}|}{(1+\varepsilon)}$, i.e., the lower bound on the burning sources, is strictly  larger than $\delta$, then it violates Equation~\ref{eq} and our guess can be increased. We stop as soon as we have $\frac{|D'_{\delta-1}|}{(1+\varepsilon)}\le \delta$. Since none of the previous guesses were successful, here we know that  $\delta^*\ge  \delta$.
 
To burn $P$, 
we first choose $D'_{\delta-1}$ as the burning sources and burn them in arbitrary order. We then keep burning another $(\delta-1)$ steps (or, stop early if all points are burnt). Since all the points are within the distance $(\delta-1)$ from some point in $D'_{\delta-1}$, all the points will be burnt. Since $|D'_{\delta-1}| \le (1+\varepsilon) \delta$, and since $\delta^*\ge  \delta$, the length of the burning sequence we compute is  $|D'_{\delta-1}|+(\delta-1) \le  (1+\varepsilon) \delta^* + \delta^* =  (2+\varepsilon) \delta^*$.
\end{proof}

\subsection{Anywhere Burning}

We   leverage the discrete unit disk cover problem to obtain a $(2+\varepsilon)$-approximation for anywhere burning. The input of a \emph{discrete unit disk cover}  problem is a set of points $P$ and a set of unit disks $\mathcal{U}$ in $\mathbb{R}^2$, and the task is to choose the smallest set $U\subseteq \mathcal{U}$ that covers all the points of $P$. There exists a PTAS for the discrete unit disk cover problem~\cite{DBLP:journals/dcg/MustafaR10}. 

We relate the discrete unit disk cover problem to anywhere burning  using the observation that there exists an optimal anywhere burning sequence  where each burning source either coincides with a given point or lies at the center of some circle determined by two or three given points. More specifically, consider a burning source $q$  with a burning radius   $r$ in an optimal anywhere burning process. Let $S$ be the set of points burned by $q$. Let $C$ be the smallest circle that covers all the points of $S$. Then we could choose a burning source at the center of $C$ instead of at $q$ and burn all points of $S$.  

A $(2+\varepsilon)$-approximation for anywhere burning problem can now be obtained by iteratively guessing the anywhere burning number using the same technique as in Section~\ref{apppb} but using an approximation to the discrete unit set cover problem. 

\begin{theorem}\label{thmab}
Given a set $P$ of points in $\mathbb{R}^2$ and an $\varepsilon{>}0$, one can compute an anywhere burning sequence in polynomial time such that the length of the sequence is at most $(2+\varepsilon)$ times the anywhere burning number of $P$.
\end{theorem}
\begin{proof} Let $P$ be the input to the anywhere burning problem. Note that Equation~\ref{eq} holds also for anywhere burning. 
We now iteratively guess the anywhere burning number $\delta$ from $1$ to $n$, where $|P|=n$. For each $\delta$, we construct   a set 
of ${{n}\choose{3}} + {{n}\choose{2}}$ disks, where each disk is of radius $\delta$ and is centered at the center of a circle determined by either two or three points of $P$. We compute a $(1+\varepsilon)$-approximation $U'_\delta$ for the  discrete unit disk cover $U_\delta$.  If $\frac{|U'_{\delta}|}{(1+\varepsilon)}$, i.e., the lower bound on the burning sources, is strictly  larger than $\delta$, then it violates Equation~\ref{eq} and our guess can be increased. We stop as soon as we have $\frac{|U'_{\delta}|}{(1+\varepsilon)}\le \delta$. Here we know that  $\delta^*\ge  \delta$.
 
To burn all the points of $P$, we first choose $U'_{\delta}$ as the burning sources and burn them in arbitrary order. We then keep burning another $(\delta-1)$ steps (or, stop early if all points are burnt). Since all the points are within the distance $(\delta-1)$ from some point in $U'_{\delta}$, all the points will be burnt. Since $|U'_{\delta}| \le (1+\varepsilon) \delta$, and since $\delta^*\ge  \delta$, the length of the burning sequence we compute is  $|U'_{\delta}|+(\delta-1) \le  (1+\varepsilon) \delta^* + \delta^* =  (2+\varepsilon) \delta^*$.
\end{proof}

\section{NP-hardness}
 



\subsection{Point Burning}\label{sec:pbh}
Consider a decision version of the point burning problem where given a set of points and an integer $b$, the task is to decide whether there is a burning sequence that burns all the points in at most $b$ steps. This decision version of the point burning problem is in NP because given a sequence of burning sources, in polynomial time one can simulate the burning process to check whether all the points are burnt. We now consider the hardness. 

  The graph burning number problem is NP-hard even for a forest of paths~\cite{DBLP:conf/waw/BonatoJR14}. To prove the NP-hardness of the point burning one can easily reduce the path forest burning problem into the point burning problem as follows. 

Let $I$ be an instance of the path forest burning problem and let $L_1,\ldots, L_t$ be the paths in $I$. We draw the vertices of each path $L_i$, $1\le i\le t$, along the x-axis in the (left-to-right) order they appear on the path with unit length distance between consecutive vertices. We ensure a gap of $(2n+1)$ units between consecutive paths, where $n$ is the number of vertices in the forest. The point burning number for the vertices of the paths is at most $n$. Since we can only burn the points (equivalently, vertices) in the point burning model, any point burning process can be seen as a graph burning and vice versa. Hence we have the following theorem.

\begin{theorem}
The point burning problem is NP-complete.
\end{theorem}

\subsection{Anywhere Burning}
\begin{figure}[h]
    \centering
\includegraphics[width=.9\textwidth]{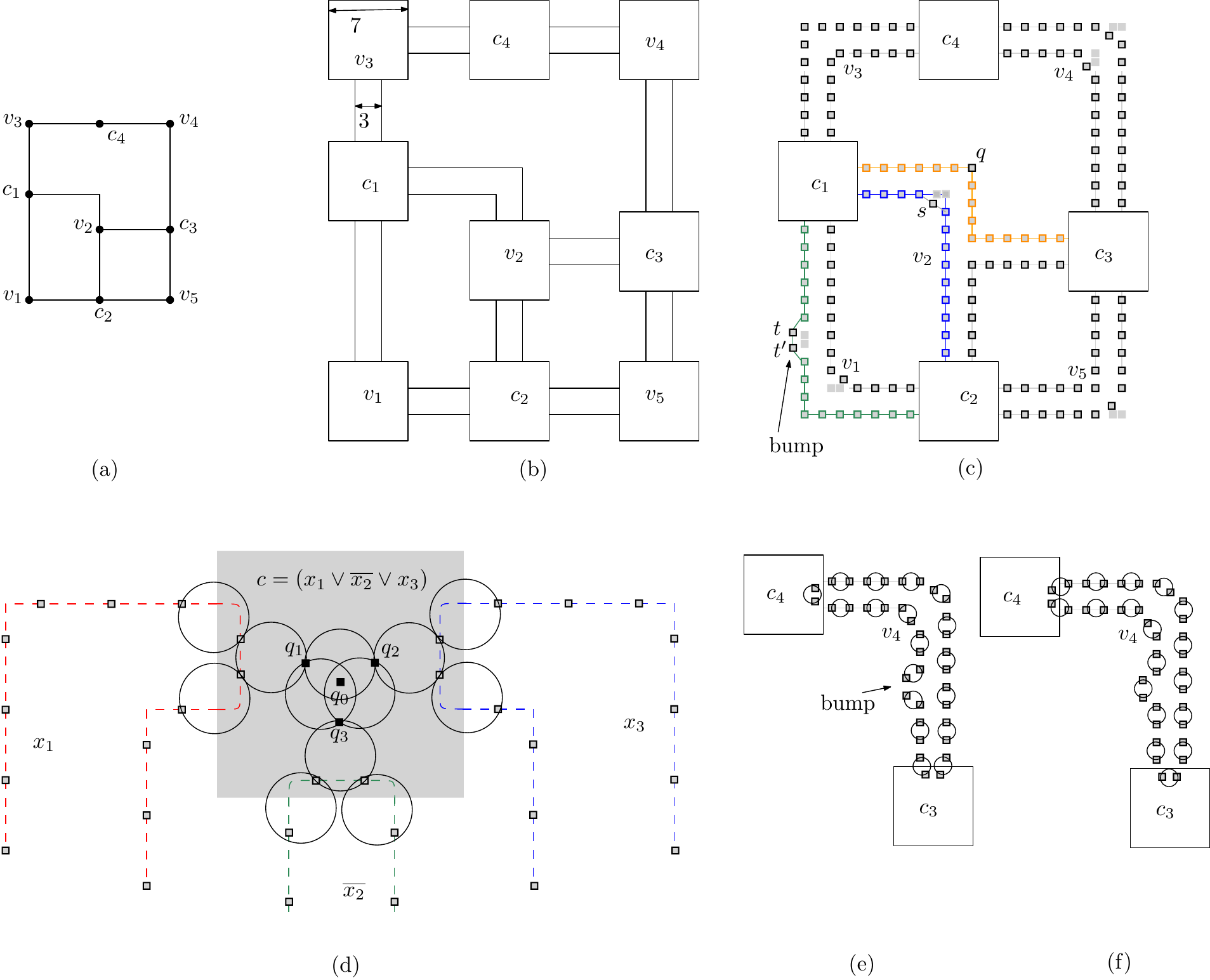}
    \caption{Illustration for the construction of the point set $P$.  }
    \label{hard}
\end{figure}

Similar to point burning, the decision version of the anywhere burning problem is in NP because given a sequence of burning sources, in polynomial time one can simulate the burning process to check whether all the points are burnt.

To show the hardness we can use almost the same hardness reduction that we used for point burning. Let $I$ be an instance of the path forest burning problem and let $P$ be the corresponding point set we constructed in Section~\ref{sec:pbh}. If the burning number for $I$ is $b$, then we can simulate the same burning process to burn all points of $P$ in $b$ steps. If $P$ admits an anywhere burning within $b$ steps, then for each burning source $q$, we  choose the nearest point $q'$ of $P$ as the burning source. By the construction of $P$, the distance between $q$ and $q'$ is at most $0.5$. Since the burning radius of $q$ is an integer, a burning source with the same radius at $q'$ will burn the same set of points as that of $q$. Therefore, if we now burn the chosen points of $P$ in the order corresponding to the anywhere burning sequence, this simulates a graph burning process on $I$ and  burns all the vertices within $b$ steps.  We thus have the following theorem.

\begin{theorem}
The anywhere burning problem is NP-complete.
\end{theorem}

\textbf{Hardness of Approximation with Bounded Burning Sources:} If we put a  restriction on the  number of burning sources that can be used, then we can modify the above hardness proof to derive an  inapproximability result  on the number of burning steps.
  
To show the hardness of approximation, we first give a different NP-hardness proof for anywhere burning. Here we   
reduce the NP-hard problem planar exactly 3-bounded  3-SAT~\cite{DBLP:journals/combinatorica/MiddendorfP93}. The input of the problem is a 3-CNF formula where each variable appears in exactly 3 clauses, each clause contains at least two and at most three literals, and the corresponding SAT graph (a graph with clauses and variables as vertices, and clause-variable incidences as edges) is planar. The task is to decide whether there exists a truth value assignment for the variables that satisfies all the clauses.  

Let $I$ be an instance of the planar exactly 3-bounded 3-SAT and let $G$ be the corresponding SAT graph. We now show how to compute a point set  $P$ and an integer $b$ such that $P$ can be burned within $b$ steps if and only if $I$ admits an affirmative solution. Our construction is inspired by an  NP-hardness reduction for the $k$-center problem~\cite{DBLP:journals/siamcomp/MegiddoS84}, but contains nontrivial details due to variable sizes for the burning radii. We will use the concept of \emph{$\beta$-disk}, which is a disk of radius $\beta$.

We first compute an orthogonal planar drawing $D$ of $G$ where the vertices are represented as grid points and edges as orthogonal polylines (Figure~\ref{hard}(a)). Every planar graph with maximum degree three has such an orthogonal planar grid drawing inside a square of side length $\lfloor n/2 \rfloor$~\cite{DBLP:journals/algorithmica/Kant96}. We then scale up the drawing and replace the vertices with squares and edges with parallel orthogonal lines (Figure~\ref{hard}(b)). We will refer to this new representation for an edge as a  \emph{tunnel}. We ensure that each square is of side length 7 units and the pair of parallel line segments of a tunnel are 3 units apart. We then replace the square for each variable by joining the tunnels incident  to it (Figure~\ref{hard}(c)).

\textbf{Creating Points for Variables and Clauses: } We now add some points along the boundary of the tunnels as follows. Let $L$ be a polygonal line (determining a side of the tunnel) from a clause to another clause  (e.g., the  orange line Figure~\ref{hard}(c)). We place points from both ends such that no three points can be covered by a $1.10$-disk. The first point is placed at one unit distance from the boundary of the square representing the clause, and then each subsequent point is placed two units apart from the previous one. If the two sequences of points from the two ends of $L$ meet at a common point  (e.g., the point $q$ in Figure~\ref{hard}(c)), then nothing else needs to be done. If the two sequences   does not meet at a common point and a bend point  is available, then we create a new point instead of creating two points that are one unit apart (e.g., the point $s$ in Figure~\ref{hard}(c)). This ensures the property that no three points can be covered by a $1.10$-disk. Note that instead of modifying a bend, one can also create a `bump' on $L$ to ensure this property, as illustrated with the points $t,t'$ in Figure~\ref{hard}(c). 

For each square representing a clause, we add 2 points for each variable incident to it and an additional 4 points $q_0,q_1,q_2,q_3$, as illustrated in Figure~\ref{hard}(d). We refer to the points $q_0,q_1,q_2,q_3$ as the \emph{clause points}. Some black unit disks are drawn to illustrate  the configuration of these  points. The key property here is that no 1.10-disk can  cover all clause  points, but if we exclude one clause point among $\{q_1,q_2,q_3\}$, then the remaining clause points can be covered using a unit disk.  
 \begin{figure*}[pt]
    \centering
\includegraphics[width=.85\textwidth]{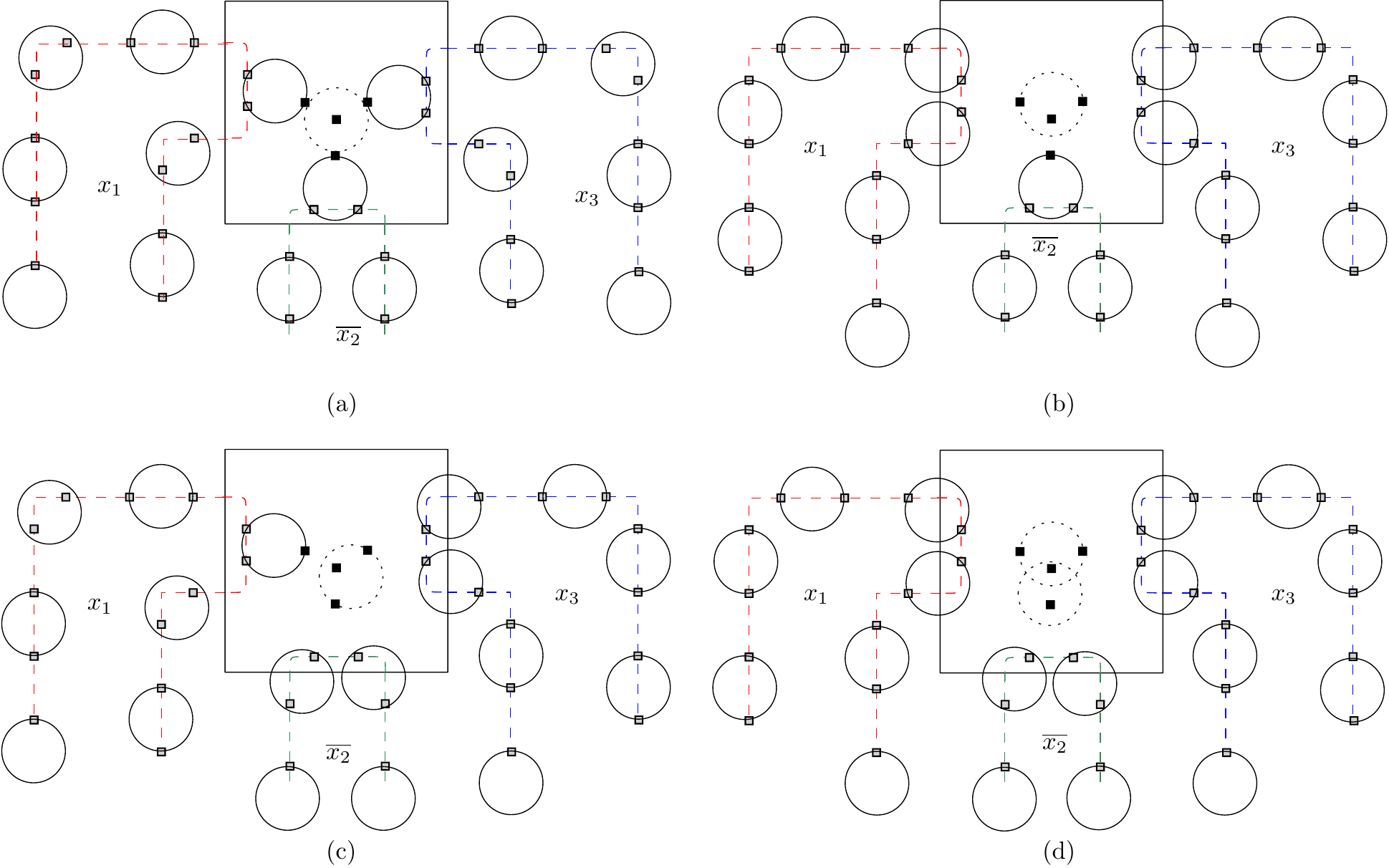}
    \caption{Illustration for the reduction. (a) $x_1=True$, $x_2=False$, $x_3=True$. (b) $x_1=False$, $x_2=False$, $x_3=False$.  (c) $x_1=True$, $x_2=True$, $x_3=False$. (d)  $x_1=False$, $x_2=True$, $x_3=False$. }
    \label{reduc}
\end{figure*}

Note that each variable now corresponds to a sequence of points forming a \emph{loop}. We create some more bumps to ensure that each variable contains an even number of points. This allows us to have two ways of covering the loop by using only unit disks by taking alternating pairs, as illustrated in Figures~\ref{hard}(e)--(f). Later, we will relate such covering to burning and if both variable-loop points inside the clause gadget are covered by the same  unit disk, then will set that literal to true. Therefore, we add some more points to ensure consistency. For example, assume that in Figures~\ref{hard}(e)--(f), the clauses $c_4$ and $c_3$ contain the literals $\overline{v_4}$ and   $v_4$, respectively.  We create a bump so that if both variable-loop points inside the gadget of $c_4$ are covered by a single unit disk, then  the two variable-loop points inside the gadget of $c_3$ will be covered by two different unit disks, and vice versa. 

Since the width and height of the drawing is of size $O(n)$, the total number of points is $O(n^2)$. We will denote by $N_v$ and $N_c$ the points that we created for the variables and clauses, respectively. 

\textbf{Creating Points to Accommodate Burning Process: } We now scale up the 
 drawing by $r$ units, where we set $r$ to be $10\left(\frac{|N_v|}{2}+\frac{|N_c|}{4}\right)$. 
 Let the resulting drawing be $D'$.  Consequently, all the above covering properties for unit disks and $1.10$-disks now hold for $r$-disks and $1.10r$-disks, respectively. 
 
 We now create $r$   points $w_i$, where $1\le i\le r$, along a horizontal line such that each point is far from the rest of the points by at least $3r$ units. 
 
 We will refer to the points created in this step as the \emph{outlier points} and denote them by $N_t$. Note that the points of $N_t$ lie  outside of $D'$. 

\textbf{From 3-SAT to Burning Number: } We now show that if the 3-SAT instance $I$ admits an affirmative solution, then the point set $(N_v\cup N_c \cup N_t)$ can be burned in $1.10r$ steps.

 In the first $0.10r$ steps we initiate $0.10r$ burning sources inside $D'$ and then initiate $r$ burning sources at the \emph{outlier points}. After this, the minimum radius of the burned area for any burning source started within $D'$ is at least $r$ and the maximum radius for such sources is $(1.10r-1)$. The burning sources inside $D'$ can be seen as $\beta$-disks where $\beta\in [r,1.10r]$. 
 
 For each true literal, we cover the corresponding two variable points and the nearest clause point by initiating a single burning source (e.g., Figures~\ref{reduc}(a)--(c)). We then burn the variable loops by initiating burning sources for alternating pairs of points. This takes $|N_v|/2$ burning sources. Since all clauses are satisfied, for each clause, at least one of the clause points from $\{q_1,q_2,q_3\}$ will be allocated to burn along with a pair of variable-loop points. Therefore, each clause now requires one burning source to ensure the burning of all its clause points. Hence the total number of burning sources we use within $D'$ is $0.10r = \left(\frac{|N_v|}{2}+\frac{|N_c|}{4}\right)$. The set $N_t$ contains $r$ points where no two of them can be covered by a $1.10r$-disk. It is straightforward to burn them in $r$ steps. Therefore, the total number of steps required is $1.10r$.

\textbf{From   Burning Number to 3-SAT: } We now show that if  the point set $(N_v\cup N_c \cup N_t)$ can be burned in $1.10r$ steps, then the 3-SAT instance $I$ admits an affirmative solution.

Since the set $N_t$ contains $r$ points where no two of them can be covered by a $1.10r$-disk, any burning sequence would need $r$ burning sources outside of $D'$. Since there are at most $1.10r$ steps to burn all the points, we are left with at most $0.10r = \left(\frac{|N_v|}{2}+\frac{|N_c|}{4}\right)$ burning sources inside $D'$. Note that none of these burning sources can have a radius larger than $(1.10r-1)$. By the construction of the variable-loop points, no three points can be covered by a $1.10r$-disk. Therefore, the variable-loop requires at least $|N_v|/2$ burning sources. If none of $q_1,q_2,q_3$ are burned along with the variable-loop points, then a clause gadget requires two burning sources (e.g., Figures~\ref{reduc}(d)).  Otherwise, each clause gadgets requires at least one burning source to ensure all clause points are burned even if $q_1,q_2,q_3$ are all burned along with the variable-loop points (e.g., Figures~\ref{reduc}(a)). Since there are $\frac{|N_c|}{4}$ clause gadgets and exactly that many burning sources remaining, each clause gadget will have exactly one from the remaining burning sources. Since a  $1.10r$-disk cannot cover all four clause points of a clause gadget, one of them must be burned together with a pair of variable-loop points. We set the corresponding literal to true. The construction of the variable-loop ensures the consistency of the truth value assignment for each variable at different clauses.

 


\textbf{Inapproximability Factor (with Bounded Burning Sources):} 
Assume that we are only allowed to initiate  $\left(\frac{|N_v|}{2}+\frac{|N_c|}{4}\right)+|N_t|$ fires. We change $r$ to be $10^\delta\left(\frac{|N_v|}{2}+\frac{|N_c|}{4}\right)$, where $\delta$ is a constant. We now can burn $(N_v\cup N_c \cup N_t)$ in $(1+10^{-\delta})r$ steps by first initiating $10^{-\delta}r = \left(\frac{|N_v|}{2}+\frac{|N_c|}{4}\right)$ burning sources inside $D'$
   and then   $r$ burning  sources to burn the points in $N_t$. Since our  reduction can be carried out with $1.10r$-disks, we can continue burning for $(1.10r-(1+10^{-\delta})r)$ more steps. Therefore, we obtain an  inapproximability factor of $\frac{1.10}{(1+10^{-\delta})}$, i.e., $(1.10-\varepsilon)$, where $\varepsilon$ can be made arbitrarily small by choosing a large value for $\delta$. 
   
    \begin{figure}[h]
    \centering
\includegraphics[width=0.3\textwidth]{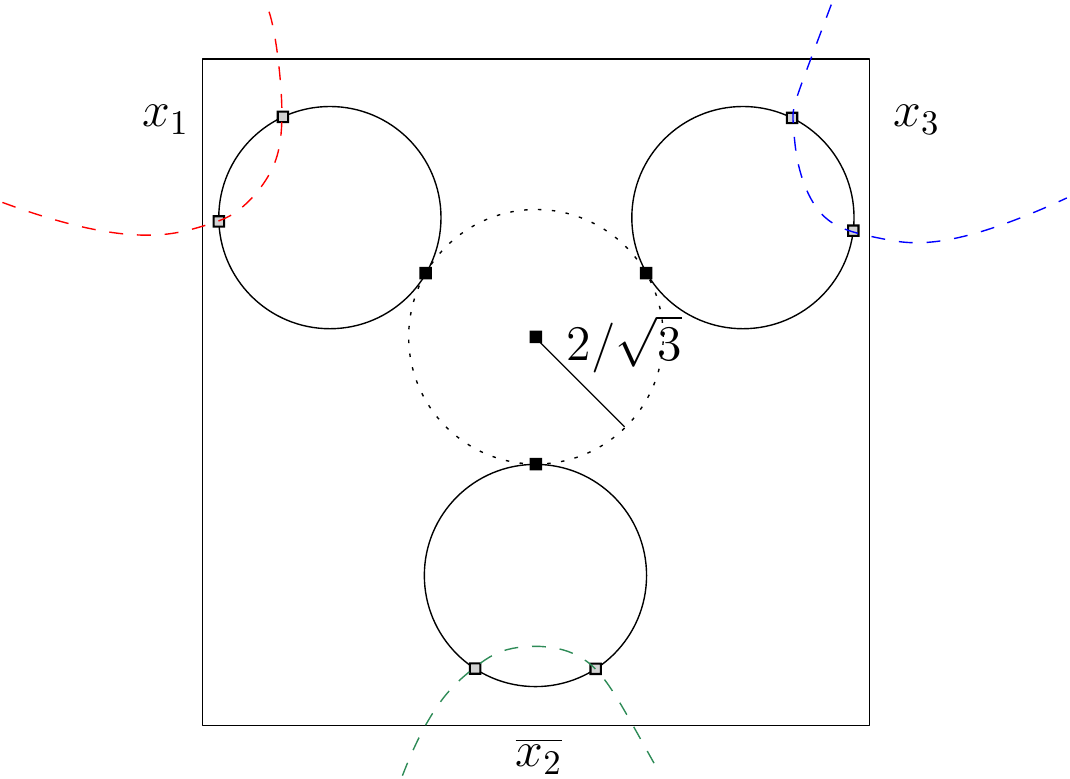}
    \caption{Illustration for the modified clause gadget. }
    \label{mod}
\end{figure}
Although for simplicity we used an orthogonal setting where variable loops  enter a clause gadget either horizontally or vertically, we could slightly change the construction using curves (similar to~\cite{DBLP:journals/siamcomp/MegiddoS84}) such that they make $120^\circ$ angles at the clause gadget (Figure~\ref{mod}). This allows us to carry out the reduction using a $\frac{2}{\sqrt{3}}$-disks and thus to have an inapproximability factor of $\frac{2}{\sqrt{3}}-\varepsilon$. 

\begin{cor}
The anywhere burning problem with a restriction on the number of burning sources that can be used is NP-hard to approximate with a factor of $\frac{2}{\sqrt{3}}-\varepsilon$, for every fixed $\varepsilon>0$.
\end{cor}

\section{Conclusion}

In this paper, we  introduced two burning processes --- point burning and anywhere burning --- to burn a set of points in the Euclidean plane. We proved that computing the burning number for these processes are NP-complete and gave approximation algorithms for them. We showed that inapproximability results can be derived for anywhere burning if only a restricted number of burning sources are allowed. Hence a natural future research  direction to explore is to design faster approximation algorithms for computing the burning number as well as to establish better inapproximability results.  

\small
\bibliographystyle{abbrv}
\bibliography{ref}





\end{document}